\documentclass[10pt]{article}

\usepackage[latin1]{inputenc}
\usepackage{amsmath,amssymb,amsthm}
\usepackage{graphicx}
\usepackage{epsfig}

\newtheorem{theorem}{Theorem}[section]

\newtheorem{corollary}[theorem]{Corollary}
\newtheorem{definition}[theorem]{Definition}

\newtheorem{lemma}[theorem]{Lemma}

\newcounter{listagem}
\newcommand{\blista}{\begin{list}{\roman{listagem})}{\usecounter{listagem}}}
\newcommand{\elista}{\end{list}}

\newcommand{\beq}{\begin{equation}}
\newcommand{\eeq}{\end{equation}}
\newcommand{\beqn}{\begin{eqnarray}}
\newcommand{\eeqn}{\end{eqnarray}}

\def\le{<}
\def\ge{>}
\newcommand{\Cl}{{C \kern -0.1em \ell}}

\newcommand{\R}{\mathbb{R}}

\begin{document}

\title{The Klein-Gordon operator on M\"obius strip domains and the Klein bottle in $\R^n$}
\author{Rolf S\"oren Krau{\ss}har \\ \\
{\small Arbeitsgruppe Algebra, Geometrie und Funktionalanalysis} \\
{\small Fachbereich Mathematik} \\ {\small Technische Universit\"at Darmstadt }\\{\small Schlo{\ss}gartenstr. 7}\\{\small 64289 Darmstadt, Germany.} \\ {\small E-mail: krausshar@mathematik.tu-darmstadt.de}} \maketitle


\begin{abstract}

In this paper we present explicit formulas for the fundamental solution to the Klein-Gordon operator on some higher dimensional generalizations of the M\"obius strip and the Klein bottle with values in distinct pinor bundles. The fundamental solution is described in terms of generalizations of the Weierstra{\ss} $\wp$-function that are adapted to the context of these geometries. 
The explicit formulas for the kernel then allow us to express all solutions to the homogeneous and inhomogeneous Klein-Gordon problem with given boundary data in the context of these manifolds. In the case of the Klein bottle we are able to describe all null solutions of the Klein-Gordon equation in terms of finite linear combinations of the fundamental solution and its partial derivatives. 
\end{abstract}


\section{Introduction}

The Klein-Gordon equation is a relativistic version of the Schr\"odinger equation. It describes the motion of a quantum scalar or pseudoscalar field, a field whose quanta are spinless particles. The Klein-Gordon equation describes the quantum amplitude for finding a point particle in various places, cf. for instance \cite{D,S}. It can be expressed in the form 
$$
(\Delta -\alpha^2-\frac{1}{c^2} \frac{\partial^2}{\partial t^2})u(x;t) = 0,
$$
where $\Delta := \sum_{i=1}^3 \frac{\partial^2}{\partial x_i^2}$ is the usual Euclidean Laplacian in $\R^3$ and $\alpha = \frac{mc}{\hbar}$. Here, $m$ represents the mass of the particle, $c$ the speed of light and $\hbar$ is the Planck number.   
\par\medskip\par
Since a long time it is well known that any solution of the Dirac equation, which describes the spinning electron, satisfies the Klein-Gordon equation. However, the converse is not true. In the time-independent case the homogeneous Klein-Gordon equation simplifies to a screened Poisson equation of the form  
$$
(\Delta -\alpha^2) u(x) = 0. 
$$
The resolution of this equation provides the first step to solve the more complicated time-dependent case. For this reason, the study of the time-independent solutions is very important. In this paper we therefore focus exclusively on the time-independent equation. 
 
As explained extensively in the literature, see for example \cite{GS2,KC} and elsewhere, the quaternionic calculus has proven to be a useful tool to study the time-independent Klein-Gordon equation. It allows us to factorize the Klein-Gordon operator elegantly by $$\Delta  - \alpha^2 =- (D -i \alpha) (D + i\alpha)$$ where $D := \sum_{i=1}^3 \frac{\partial }{\partial x_i} e_i$ is the Euclidean Dirac operator associated to the spatial variable $x=x_1 e_1+x_2 e_2 + x_3 e_3$. In the quaternionic calculus the basis elements $e_1,e_2,e_3$ behave like algebraically the quaternionic basic units ${\bf i}$, ${\bf j}$ and ${\bf k}$. The study of the solutions of the original scalar second order equation is thus reduced to study vector valued eigensolutions to the first order Dirac operator associated to purely imaginary eigenvalues. For eigensolutions of the first order Euclidean Dirac operator it was possible to develop a powerful higher dimensional version of complex function theory, see for instance \cite{GS2,KC,Zhenyuan2,Ry94}. By means of these function theoretical methods it was possible to set up fully analytic representation formulas for the solutions to the homogeneous and inhomogeneous Klein-Gordon in the three dimensional Euclidean space in terms of quaternionic integral operators. 
\par\medskip\par 
In our recent papers \cite{KraKG,ConKraMMAS2009}, we have shown how these techniques can be carried over to the context of conformally flat oriented cylinders and tori in $\R^3$, and more generally in $\R^n$. The fundamental solution of the Klein-Gordon operator has been constructed from elementary spinor valued multiperiodic solutions of the Klein-Gordon operator. 
\par\medskip\par
The aim of this paper is to construct with analogous pseudo-multiperiodic series constructions the fundamental solution of the  Klein-Gordon operator on some important examples non-orientable conformally flat manifolds. In the context of this paper we focus on higher dimensional generalizations of the M\"obius strip and the Klein bottle. These manifolds play an important role in general quantum field theory and relativity theory, but also in string and $M$-theory, see for instance \cite{BBS}.
\par\medskip\par
In contrast to the case of the oriented tori and cylinder that we considered earlier, which are examples of spin manifolds, we cannot construct the fundamental solution of the Klein-Gordon equation in terms of spinor valued sections that are in the kernel of $D - i\alpha$. First of all due to the lack of orientability, we cannot construct spinor bundles over these manifolds. Secondly, it is not possible either to construct non-vanishing solutions in the class Ker $D-i\alpha$ in $\R^n$ that have the additional pseudo periodic property to descend properly to these manifolds. A successful way is to start directly from special classes of harmonic functions that take values in bundles of the $^+Pin(n)$ group or $^-Pin(n)$ group. As described for instance in \cite{BD,BRS,KT} the consideration of pin structures is the most logical analogue to spin structures in the non-oriented cases. Pin structures are crucially applied in unoriented superstring theory where non-orientable worldsheets are considered, cf. for instance \cite{BM}. 

By means of special classes of pseudo-multiperiodic harmonic functions we develop series representations for the Green's kernel of the Klein-Gordon operator for some $n$-dimensional generalizations of the M\"obius strip and the Klein bottle with values in different pin bundles. These functions represent a generalization of the Weierstra{\ss} $\wp$-function to the context of these geometries. 

As applications we present some integral formulas that provide us with the basic stones for doing harmonic analysis in this geometrical context. This paper provides an extension to our prior works \cite{KraNonorientable} and \cite{KRV} in which we developed analogous formulas for the Laplacian and for the Schr\"odinger operator on these manifolds. 

\par\medskip\par

The case of the Klein bottle has interesting particular features. As a consequence of the compactness of the Klein bottle we can prove that every solution of the Klein-Gordon operator having atmost unsessential singularities can be expressed as a finite linear combination of the fundamental solution and a finite amount of its partial derivatives. The only entire solution of the Klein-Gordon equation on the Klein bottle is the function $f \equiv 0$.    

\section{Pin structures on conformally flat manifolds}

Conformally flat manifolds are $n$-dimensional Riemannian manifolds that possess atlases whose transition functions are conformal maps in the sense of Gauss. For $n > 3$ the set of conformal maps coincides with the set of M\"obius transformations. In the case $n=2$ the sense preserving conformal maps are exactly the holomorphic maps. So, under this viewpoint we may interpret conformally flat manifolds as higher dimensional generalizations of holomorphic Riemann surfaces. On the other hand, conformally flat manifold are precisely those Riemannian manifolds which have a vanishing Weyl tensor. 

\par\medskip\par

As mentioned for instance in the classical work of N. Kuiper \cite{Kuiper}, concrete examples of conformally flat orbifolds can be constructed by factoring out a simply connected domain $X$ by a Kleinian group $\Gamma$ that acts discontinuously on $X$. In the cases  where $\Gamma$ is torsion free, the topological quotient $X/\Gamma$, consisting of the orbits of a pre-defined group action $\Gamma \times X \to X$, is endowed with a differentiable structure. We then deal with examples of conformally flat manifolds. 
    
\par\medskip\par
In the case of oriented manifolds it is natural to consider spin structures. In the non-oriented case, this is not possible anymore. However, one can consider pin structures instead. For details about the description of pin structures on manifolds that arise as quotients by discrete groups. we refer the reader for instance to \cite{BD}. See also \cite{BRS} and \cite{KT} where in particular the classical M\"obius strip and the classical Klein bottle has been considered.
\par\medskip\par  
A classical way of obtaining pin structures for a given Riemannian manifold is to look for a lifting of the principle bundle associated to the orthogonal group $O(n)$ to a principle bundle for the pin groups $^{\pm}Pin(n)$. As described in the above cited works, the group $^+Pin(n):=Pin(n.0)$ is associated to the Clifford algebra $Cl_{n,0}$ of positive signature $(n,0)$. The Clifford algebra $Cl_{n,0}$ is defined as the free algebra modulo the relation $x^2= q_{n,0}(x)$ ($x \in \R^n$) where $q_{n,0}$ is the quadratic form defined by $q_{n,0}(e_i)=+1$ for all basis vectors $e_1,\ldots,e_n$ of $\R^n$. For particular details about Clifford algebras and their related classical groups we also refer the reader to \cite{p}. Next we recall that the group 
$^-Pin(n) := Pin(0,n)$ is associated to the Clifford algebra $Cl_{0,n}$ of negative signature $(0,n)$. Here the quadratic form $q_{n,0}$ is replaced by the quadratic form $q_{0,n}$ defined by $q_{0,n}(e_i) =-1$ for all $i=1,\ldots,n$.  
Topologically both groups are equivalent, however algebraically they are not isomorphic, cf. for example~\cite{KT}.  The more popular $Spin(n)$ group is a subgroup of $^{\pm}Pin(n)$ of index $2$. Here we have $Spin(n):=Spin(0,n) \cong Spin(n,0)$. $Spin(n)$ consists exactly of those matrices from $^{\pm}Pin(n)$ whose determinant equals $+1$. 
The groups $^{\pm}Pin(n)$ double cover the group $O(n)$. 

So there are surjective homomorphisms $^{\pm}\theta: ^{\pm}Pin(n)\rightarrow ^{\pm}Pin(n)$ with kernel ${\mathbb{Z}}_{2}=\{\pm 1\}$. Adapting from Appendix C of \cite{pp}, where spin structures have been discussed, this homomorphism gives rise to a choice of two local liftings of the principle $O(n)$ bundle to a principle $^{\pm}Pin(n)$ bundle. The number of different global liftings is given by the number of elements in the cohomology group $H^{1}(M,{\mathbb{Z}}_{2})$. See \cite{lm} and elsewhere for details. These choices of liftings give rise to different pinor bundles over $M$. We shall explain their explicit construction on the basis of the examples that we consider in the following two sections.

\section{Higher dimensional M\"obius strips}

In all that follows $\{e_1,\ldots,e_n\}$ stands for the standard basis in $\R^n$. 
\par\medskip\par
Let $k \in \{1,\ldots,n\}$ and $\underline{x}= x_1 e_1 + \cdots + x_k e_k$ be a reduced vector from $\mathbb{R}^k$. Let 
$\Lambda_k:= \mathbb{Z} e_1 + \cdots + \mathbb{Z} e_k$ be the $k$-dimensional orthonormal lattice. 
Suppose further that $\underline{v}:= m_1 e_1 + \cdots + m_k e_k$ is a vector from that lattice $\Lambda_k \subset \mathbb{R}^k$.
\par\medskip\par 
Now it is natural to introduce higher dimensional analogoues of the M\"obius strip by the factorization
$$
{\cal{M}}_k^{-} = \mathbb{R}^n/\sim
$$
where $\sim$ is defined by the map $$(\underline{x}+\underline{v},x_{k+1},\ldots,x_{n-1},x_n) \mapsto (x_1,\ldots,x_k,x_{k+1},\ldots,x_{n-1},{\rm sgn}(\underline{v})x_n).$$
Here, for $\underline{v} = m_1 e_1 + \cdots + m_k e_k$ we write
${\rm sgn}(\underline{v}) = \left\{ \begin{array}{cc} 1 & {\rm if}\; \underline{v} \in 2 \Lambda_k \\ -1 & {\rm if}\; \underline{v} \in \Lambda_k \backslash 2 \Lambda_k. \end{array} \right. $

The manifold ${\cal{M}}_k^{-}$ consists of the orbits of the group action $\Lambda_k \times \R^n \to \R^n$, where the action is defined by 
$$
\underline{v} \circ x := (\underline{x}+\underline{v},x_{k+1},\ldots,x_{n-1},{\rm sgn}(\underline{v}) x_n).
$$
The classical M\"obius strip is obtained in the case $n=2,k=1$ where the pair $(x_1+1,x_2,X)$ is mapped to $(x_1,-x_2,X)$ after one period. Due to the switch of the minus sign in the $x_n$-component we really deal in that context with non-orientable manifolds. The manifolds ${\cal{M}}_k^-$ are no spin manifolds.

\par\medskip\par

However, analogously to the case of a spin manifold, we can set up several distinct pinor bundles on these manifolds. An extensive algebraic description of all pin structures on the classical M\"obius strip can be found for instance in \cite{BRS}. A simple way to obtain different pin structures is to consider certain decompositions of the period lattices, as we did for the oriented cylinders and tori in our previous works, cf. e.g. \cite{KraRyan2}. 
The decomposition of the lattice $\Lambda_k$ into the direct sum of the sublattices $\Lambda_l:= \mathbb{Z} e_1 + \cdots + \mathbb{Z} e_l$ and $\Lambda_{k-l}:= \mathbb{Z} e_{l+1} + \cdots + \mathbb{Z} e_k$ gives rise to $2^k$ different pinor bundles, namely  by mapping the tupel $$(\underline{x}+\underline{v},x_{k+1},\ldots,x_n,X) \; {\rm to}\; (\underline{x}, x_{k+1},...,x_{n-1},{\rm sgn}(v) x_n,(-1)^{m_1+\cdots+m_k}X).$$
Here $X \in Cl_n(\mathbb{C})$, where $Cl_n(\mathbb{C})$ stands for the complexification of the Clifford algebra $Cl_{n,0}$ or $Cl_{0,n}$, respectively. 

\par\medskip\par

In order to describe the fundamental solution of the time-independent Klein-Gordon operator on the manifolds ${\cal{M}}_k^-$ we recall from standard literature (see for instance \cite{Zhenyuan2}) that the fundamental solution to the Klein-Gordon operator $\Delta-\alpha^2$ in Euclidean flat space $\R^n$ is given by 
\begin{equation}
E_{\alpha}(x) = - \frac{i \pi}{2 \omega_n \Gamma(n/2)} (\frac{1}{2}i\alpha)^{\frac{1}{2}n-1} |x|^{1-\frac{1}{2}n}
H_{\frac{n}{2}-1}^{(1)}(i\alpha |x|).
\end{equation}
In this formula, $\omega_n$ stands for the surface measure of the unit sphere in $\R^n$ while $H_{m}^{(1)}$ denotes the first Hankel function with parameter $m$. Furthermore, we choose the root of $\alpha^2$ such that $\alpha > 0$. Since the Hankel function $H_{\frac{n}{2}-1}^{(1)}$ has a point singularity of the order $\frac{1}{2}n-1$ at the origin, the function $E_{\alpha}(x)$ has a point singularity of order $n-2$ at the origin. This is a consequence of the prefactor $|x|^{1-\frac{1}{2}n}$. 

\par\medskip\par

We now shall see that we can construct the Green's function of the Klein-Gordon operator on ${\cal{M}}_k^-$ by summing the fundamental solution $E_{\alpha}(x)$ over the period lattice $\Lambda_k$ in a way involving the appropriate minus sign in the $x_n$-component in exactly those terms of the series which are associated to lattice points where ${\rm sgn}(\underline{v}) = -1$. To proceed in this direction we first prove:
\begin{lemma}\label{conv}
Let $\alpha \in \R \backslash \{0\}$. The series 
\begin{equation}
\wp_{\alpha}^{{\cal{M}}_k^-}(x):=\sum\limits_{\underline{v} \in \Lambda_k}  E_{\alpha}(\underline{x}+\underline{v},x_{k+1},\ldots,x_{n-1},{\rm sgn}(\underline{v})x_n)
\end{equation}
converges uniformely in each compact subset of $\R^n\backslash \Lambda_k$.  
\end{lemma}
\begin{proof}
To prove the convergence we use the well-known asymptotic formula 
\begin{equation}\label{asy}
H_m^{(1)}(\alpha |x|_2) \sim \sqrt{\frac{2}{\pi
|\alpha||x|_2}} e^{i \alpha|x|_2},\quad {\rm for}\;|x|\;{\rm large}, 
\end{equation}
also used in \cite{ConKraMMAS2009}, which holds for all positive parameters $m \in \frac{1}{2} \mathbb{N}$ whenever $|x|$ is sufficiently large. From formula (\ref{asy}) we obtain the asymptotic estimate 
$$
|E_{\alpha}(x)| \le c \frac{e^{-\alpha|x|}}{|x|^{(n-1)/2}}
$$ 
which then holds for sufficiently large values for $|x|$, where $c$ is a real constant. 

Next we decompose the lattice $\Lambda_k$ in terms of the union of the sets  $\Lambda_k = \bigcup_{m=0}^{+\infty} \Omega_m$ where
$$
\Omega_m := \{ \underline{v} \in \Lambda_k \mid |\underline{v}|_{max} = m\}.$$ 
$|\cdot|_{max}$ denotes the usual maximum norm of a vector. In order to avoid ambiguities we use the notation $|\cdot|_2$ for the Euclidean norm within this proof. 
We further consider the following subsets of this lattice $L_m := \{
\underline{v} \in \Lambda_k \mid |\underline{v}|_{max} \le m\}$. Obviously, the
set $L_m$ contains exactly $(2m+1)^n$ points. Hence, the
cardinality of $\Omega_m$ is $\sharp \Omega_m = (2m+1)^n -
(2m-1)^n$. The Euclidean distance between the set $\Omega_{m+1}$
and $\Omega_m$ has the value $d_m := dist_2(\Omega_{m+1},\Omega_m)
= 1$.  
\par\medskip\par
Now let us consider an arbitrary compact subset ${\cal{K}} \subset \R^n$.
Then there exists a positive $r \in \R$ such that all $x \in
{\cal{K}}$ satisfy $|x|_{max} \le |x|_2 < r$.
Suppose now that $x$ is a point of ${\cal{K}}$. To show the 
normal convergence of the series, we can leave out without loss of
generality a finite set of lattice points. We consider without loss of
generality only the summation over those lattice points that
satisfy $|\underline{v}|_{max} \ge [r]+1$. In view of $|x +
\underline{v}|_2 \ge |\underline{v}|_2 - |x|_2 \ge
|\underline{v}|_{max}-|x|_2 = m - |x|_2 \ge m - r$ we
obtain
\begin{eqnarray*}
& & \sum\limits_{m=[r]+1}^{+\infty} \sum\limits_{\omega \in \Omega_m} |E_{\alpha}(\underline{x}+\underline{v},x_{k+1},\ldots,x_{n-1},{\rm sgn}(\underline{v}))|_2\\
& \le & c \sum\limits_{m=[r]+1}^{+\infty} \sum\limits_{\omega \in \Omega_m}
 \frac{e^{-|\alpha||\underline{x}+\underline{v}+ x_{k+1}e_{k+1} + \cdots + {\rm sgn}(\underline{v}) x_n e_n|_2}}
 {\sqrt{(|\underline{x}+\underline{v}+x_{k+1} + x_{k+1}e_{k+1} + \cdots + {\rm sgn}(\underline{v}) x_n e_n|_2)^{n-1}}}\\
& \le & c \sum\limits_{m=[r]+1}^{+\infty} [(2m+1)^n - (2m-1)^n]
\frac{e^{-|\alpha|(r-m)}}{(m-r)^{(n-1)/2}},
\end{eqnarray*}
where $c$ is an appropriately chosen positive real constant,
because $m - r \ge [r]+1-r > 0$. This sum clearly is absolutely uniformly convergent. Hence,
the series $\wp_{\alpha}^{{\cal{M}}_k^-}(x)$, 
which can be written as
$$
\wp_{\alpha}^{{\cal{M}}_k^-}(x) :=
\sum\limits_{m=0}^{+\infty}\sum\limits_{\underline{v} \in \Omega_m}
E_{\alpha}(\underline{x} + \underline{v},x_{k+1},\ldots,x_{n-1},{\rm sgn}(\underline{v})x_n),
$$
converges normally on $\R^n \backslash \Lambda_k$.
\end{proof}
{\bf Remark}. The case $\alpha=0$ is the harmonic case treated in \cite{KraNonorientable}.
\par\medskip\par
Next we may establish that 
\begin{lemma}\label{reg}
The function $\wp_{\alpha}^{{\cal{M}}_k^-}(x)$ is an element from Ker $\Delta-\alpha^2$ on $\R^n\backslash \Lambda_k$.  
\end{lemma}
\begin{proof}
The main point is to observe that each term 
$$
E_{\alpha}(\underline{x} + \underline{v},x_{k+1},\ldots,{\rm sgn}(\underline{v})x_n)
$$
is an element from  Ker $\Delta-\alpha^2$ in $\mathbb{R}^n\backslash\{(\underline{v},0)\}$. This is trivial for $\underline{v}=\underline{0}$, since $E_{\alpha}(x)$ is the fundamental solution of $\Delta-\alpha^2$. Next, since we differentiate twice to each variable (also w.r.t. $x_n$), the term 
$E_{\alpha}(\underline{x},\ldots,x_{n-1},\pm x_n)$ turns out to be an element from Ker $\Delta-\alpha^2$, too. Finally, due to the invariance of  $\Delta-\alpha^2$ under translations of the form $x \mapsto x + \underline{v}$, each term  $E_{\alpha}(\underline{x} + \underline{v},x_{k+1},\ldots,{\rm sgn}(\underline{v})x_n)$ then is an element from that kernel, due to the previous arguments. 
Now we may conclude by Weierstra{\ss}' convergence theorem that the entire series $\wp_{\alpha}^{{\cal{M}}_k^-}(x)$ is an element from that kernel in all points of $\R^n \backslash \Lambda_k$, due to the normal convergence.   
\end{proof}
The next step is to observe 
\begin{lemma}\label{periodic}
The function $\wp_{\alpha}^{{\cal{M}}_k^-}(x)$ satisfies
\begin{equation}
\wp_{\alpha}^{{\cal{M}}_k^-}(\underline{x}+\underline{\eta},x_{k+1},\ldots,x_n) = \wp_{\alpha}^{{\cal{M}}_k^-}(\underline{x},x_{k+1}.\ldots,x_{n-1},{\rm sgn}(\underline{\eta})x_n) \quad\quad \forall \underline{\eta} \in \Lambda_k.
\end{equation}
\end{lemma}
\begin{proof}
Take an arbitrary $\underline{\eta} \in \Lambda_k$. Then 
\begin{eqnarray*}
& & \wp_{\alpha}^{{\cal{M}}_k^-}(\underline{x}+\underline{\eta},x_{k+1},\ldots,x_n)\\ &=&  \sum\limits_{\underline{v} \in \Lambda_k}  E_{\alpha}(\underline{x}+\underline{\eta}+\underline{v},x_{k+1},\ldots,x_{n-1},{\rm sgn}(\underline{v})x_n)\\
&=& \sum\limits_{\underline{v} \in \Lambda_k}  E_{\alpha}(\underline{x}+(\underline{v}+\underline{\eta}),x_{k+1},\ldots,x_{n-1},{\rm sgn}( \underline{\eta}){\rm sgn}(\underline{v}+\underline{\eta})x_n)\\
&=& \sum\limits_{\underline{v'} \in \Lambda_k}  E_{\alpha}(\underline{x}+\underline{v'},x_{k+1},\ldots,x_{n-1},{\rm sgn} (\underline{\eta}){\rm sgn}(\underline{v'})x_n)\\
&=& \wp_{\alpha}^{{\cal{M}}_k^-}(\underline{x},\ldots,x_{n-1},{\rm sgn}(\underline{\eta})x_n),
\end{eqnarray*}
where we applied a direct rearrangement argument. 
\end{proof}
Now let $p_-$ be the well-defined projection map from the universal covering space $\R^n$ down to the manifold ${\cal{M}}_k^- = \R^n/\sim$. 
As a consequence of Lemma~\ref{periodic} we may now conclude that the function $\wp_{\alpha}^{{\cal{M}}_k^-}(x)$ descends to a well-defined pinor section ${\wp'}_{\alpha}^{{\cal{M}}_k^-}(x'):= p_{-}(\wp_{\alpha}^{{\cal{M}}_k^-}(x))$ on the manifold ${\cal{M}}_k^-$ which takes its values in the trivial pinor bundle. Here, we use the notation $x':= p_{-}(x) = x \;{\rm mod}\; \sim$. The projection map $p'$ induces canonically a  Klein-Gordon operator ${\Delta'}_{\alpha} = p_-(\Delta-\alpha^2)$ on ${\cal{M}}_k^-$. As a consequence of Lemma~\ref{reg}, the section ${\wp'}_{\alpha}^{{\cal{M}}_k^-}(x')$ then is in Ker ${\Delta'}_{\alpha}$. Finally, we are in position to establish the main result of this paragraph 
\begin{theorem}  
Consider the decomposition of the lattice $\Lambda_k = \Lambda_{l} \oplus \Lambda_{k-l}$ for some $l \in \{1,\ldots,k\}$ and write a lattice point $\underline{v} \in \Lambda_k$ in the form $\underline{v} = m_1 e_1 + \cdots+ m_l e_l + m_{l+1}e_{l+1} + \cdots + m_k e_k$ with integers $m_1,\ldots,m_k \in \mathbb{Z}$.   
Let ${\cal{E}}^{(q)}$ be the specific pinor bundle on ${\cal{M}}_k^-$ defined by the map
$$(\underline{x}+\underline{v},x_{k+1},\ldots,x_n,X) \; {\rm to}\; (\underline{x}, x_{k+1},...,x_{n-1},{\rm sgn}(v) x_n,(-1)^{m_1+\cdots+m_k}X).$$ The fundamental solution to the Klein-Gordon operator on ${\cal{M}}_k^-$ for pinor sections with values in the pinor bundle ${\cal{E}}^{(q)}$ can be expressed by 
\begin{equation}
{E'}_{\alpha,q}(x') = p_-\Bigg(\sum\limits_{\underline{v}   \in \Lambda_l \oplus \Lambda_{k-l}} (-1)^{m_1+\cdots+m_l} E_{\alpha}(\underline{x}+\underline{v},x_{k+1},\ldots,x_{n-1},{\rm sgn}(v)x_n)  \Bigg).
\end{equation}
\end{theorem} 
In the case of the trivial bundle we shall simply write ${E'}_{\alpha}(x')$ in what follows.  
\begin{proof}
With the modification of the appropriate minus sign induced by the factor $(-1)^{m_1+\cdots+m_l}$ in each term of the series we ensure in combination with the previous lemmas that ${E'}_{\alpha}(x')$ actually is a well-defined section on ${\cal{M}}_k^-$ and it takes values in the chosen pinor bundle ${\cal{E}}^{(q)}$. With analogous arguments as used in the proof of Lemma~\ref{reg}, this section is in the kernel of the Klein-Gordon operator. Since the expression $E_{\alpha}(x-y)$ is the Green's kernel of the usual Klein-Gordon operator in the Euclidean flat space $\R^n$, its canonical projection ${E'}_{\alpha}(x')$ takes over this role on the manifold ${\cal{M}}_k^-$. 
It reproduces each section in the kernel of $\Delta'_{\alpha}$ by the Green's integral.    
\end{proof}
We can say more.  These explicit formulas for the fundamental solutions allow us to express the solutions of the homogeneous and inhomogeneous Klein-Gordon problem with and without boundary conditions.
\par\medskip\par
Let us suppose that $D' \subset {\cal{M}}_k^-$ is a bounded domain with a $C^2$-boundary. Let $\nu$ denote the unit normal vector to the boundary $\partial D$ directed into the exterior of $D'$. Assume further that $u'$ is a $C^2$-section in the kernel of ${\Delta'}_{\alpha} = p_-(\Delta-\alpha^2)$ with values in the pinor bundle ${\cal{E}}^{(q)}$. Since ${\Delta'}_{\alpha}$ is a scalar operator, we may restrict to consider each scalar component function ${u_j}'$  of the pinor valued section $u'$ separately. 

Suppose that $u$ has a normal derivative on $\partial D'$ in the sense that the limit 
$$
\frac{\partial {u_j}'}{\partial \nu}(y') = \lim\limits_{h \to 0^+} \nu(y') \cdot {\rm grad}\; {u_j}'(y'-h\nu(y')),\quad y' \in \partial D' 
$$ 
exists uniformely on $\partial D'$ for each scalar component function ${u_j}'$. Then we can infer by classical harmonic analysis arguments that each scalar component function satisfies
\begin{equation}\label{homcomp}
{u_j}'(x') = \int\limits_{\partial D'} \Big[ \frac{\partial {u_j}'}{\partial \nu}(y'){E'}_{\alpha,q}(x'-y') - {u_j}'(y')
\frac{\partial {E'}_{\alpha,q}(x'-y')}{\partial \nu(y')}\Big] dS(y).
\end{equation}
Therefore, the whole pinor valued section $u'$ satisfies 
\begin{equation}\label{hom}
{u'}(x') = \int\limits_{\partial D'} \Big[ \frac{\partial {u'}}{\partial \nu}(y'){E'}_{\alpha,q}(x'-y') - {u'}(y')
\frac{\partial {E'}_{\alpha,q}(x'-y')}{\partial \nu(y')}\Big] dS(y).
\end{equation}
Next, as a further consequence we can directly present the solution of the inhomogeneous Klein-Gordon problem $\Delta'_{\alpha} u' = f'$ in $D'$ with no-boundary condition in terms of the convolution integral over the fundamental section ${E'}_{\alpha,q}$. Under the weaker condition that $f'$ is an ${\cal{E}}^{(q)}$ valued pinor section belonging to $W^{p}_{2}(D')$ for some $p > 0$ we have   
\begin{equation}\label{inh}
u'(x') = \int\limits_{D'} {E'}_{\alpha,q}(x'-y') f'(y') dV(y').
\end{equation}  
Here, and in the following $W^{p}_2(D')$ stands for the Sobolev space consisting of the $p$-times weakly differentiable ${\cal{E}}^{(q)}$-valued sections belonging to $L_2(D')$. 
\par\medskip\par
In view of the linearity of the Klein-Gordon operator, we obtain from combining the latter statement with a weaker version of the previous one where we work in the same Sobolev spaces as in \cite{GS2}, the following statement. 
\begin{theorem} (Inhomogeneous Klein-Gordon boundary value problem on ${\cal{M}}_k^-$). 
Let $D' \subset {\cal{M}}_k^-$ be a bounded domain with a $C^2$-boundary. Suppose that $f' \in W^{p}_2(D')$ and that $g' \in W^{p+3/2}_2(\partial D')$. Then the boundary value problem $\Delta'_{\alpha} u' = f'$ in $D'$ with $u'|_{\partial D'} = g'$ has a unique solution in $W^{p+2,loc}_2$ of the form  
\begin{eqnarray*}
u'(x') &=& \int\limits_{D'} {E'}_{\alpha,q}(x'-y') f'(y') dV(y') \\&+& \int\limits_{\partial D'} \Big[ \frac{\partial u'}{\partial \nu}(y'){E'}_{\alpha,q}(x'-y') - u'(y')
\frac{\partial {E'}_{\alpha,q}(x'-y')}{\partial \nu(y')}\Big] dS(y).
\end{eqnarray*}
\end{theorem} 
 
\section{Higher dimensional generalizations of the Klein bottle}
Now we turn to discuss analogous constructions for higher dimensional generalizations of the Klein bottle. To leave it simple we consider an $n$-dimensional normalized lattice of the form $\Lambda_n:=\mathbb{Z} e_1 + \cdots + \mathbb{Z} e_n$.   
\par\medskip\par
We introduce higher dimensional generalization of the classical Klein bottle by the factorization
$$
{\cal{K}}_n:=\mathbb{R}^n/\sim^*
$$
where $\sim^*$ is now defined by the map
$$
(\underline{x} + \sum_{i=1}^{n-1} m_i e_i +(x_n + m_n) e_n) \mapsto(x_1,\cdots,x_{n-1},(-1)^{m_n} x_n).
$$
The manifolds ${\cal{K}}_n$ can be described as the set of orbits of the group action $\Lambda_n \times \R^n \to \R^n$ where the action is now defined by 
$$
v \circ x := (\sum\limits_{i=1}^{k-1} x_i e_i+\sum_{i=1}^{k-1}m_i e_i+((-1)^{m_k} x_k + m_k) e_k,x_{k+1},\ldots,x_{n-1}, x_n),
$$
where $v = m_1 e_1 + \cdots + m_n e_n$ is a lattice point from $\Lambda_n$. 
Here, and in the remaining part of this section, $\underline{x}$ denotes a shortened vector in $\mathbb{R}^{n-1}$. In the case $n=2$ we obtain the classical Klein bottle. Notice that in contrast to the M\"obius strips treated in the previous section, here the minus sign switch occurs in one of the component on which the period lattice acts, too. As for the M\"obius strips we can again set up distinct pinor bundles. See also \cite{BRS} where pin structures of the classical four-dimensional Klein-bottle have been considered.  

By decomposing the complete $n$-dimensional lattice $\Lambda_n$ into a direct sum of two sublattices $\Lambda_n = \Omega_l \oplus \Lambda_{n-l}$ we can again construct $2^n$ distinct pinor bundles by considering the maps
$$
(\underline{x} + \sum_{i=1}^{n-1} m_i e_i, x_n + m_n,X) \mapsto(x_1,\cdots,x_{n-1},(-1)^{m_n} x_n,(-1)^{m_1+\cdots+m_l}X).
$$
By similar arguments as applied in the previous section we may establish 
\begin{theorem}  
Consider the decomposition of the lattice $\Lambda_n = \Lambda_{l} \oplus \Lambda_{n-l}$ for some $l \in \{1,\ldots,n\}$ and write a lattice point $v \in \Lambda_n$ in the form $v = m_1 e_1 + \cdots+ m_l e_l + m_{l+1}e_{l+1} + \cdots + m_n e_n$ with integers $m_1,\ldots,m_n \in \mathbb{Z}$.   

Let ${\cal{E}}^{(q)}$ be the pinor bundle on ${\cal{K}}_n$ defined by the map $$
(\underline{x} + \sum_{i=1}^{n-1} m_i e_i, x_n + m_n,X) \mapsto(x_1,\cdots,x_{n-1},(-1)^{m_n} x_n,(-1)^{m_1+\cdots+m_l}X).
$$
The fundamental solution of the Klein-Gordon operator on ${\cal{K}}_n$ (induced by $p_*(\Delta-\alpha^2)$ ) for sections with values 
in the pinor bundle ${\cal{E}}^{(q)}$ can be expressed by 
\begin{equation}
{E'}_{\alpha,q}(x') = p_*\Bigg(\sum\limits_{v   \in \Lambda_l \oplus \Lambda_{n-l}} (-1)^{m_1+\cdots+m_l} E_{\alpha}(\sum_{i=1}^{n-1} (x_i + m_i) e_i,(-1)^{m_n}x_n+x_m)  \Bigg),
\end{equation}
where $p_*$ now denotes the projection from $\R^n$ to ${\cal{K}}_n = \R^n/\sim^*$. The symbol $'$ represents the image under $p_*$ up from now.
\end{theorem} 
To the proof, without loss of generality we consider the trivial bundle, as the arguments can easily be adapted to the other bundles that we also considered, namely by taking into account the parity factor $(-1)^{m_1+\cdots+m_l}$. Completely analogously to the proof of Lemma~\ref{conv} we may establish the normal convergence of the series 
\begin{equation}
\wp_{\alpha}^{{\cal{K}}_n}(x) := \sum\limits_{v   \in \Lambda_{n}} E_{\alpha}(\sum_{i=1}^{n-1} (x_i + m_i) e_i,(-1)^{m_n}x_n+x_m)
\end{equation}
in $\R^n \backslash \Lambda_n$. With the same argumentation as in Lemma~\ref{reg} we may establish that the function $\wp_{\alpha}^{{\cal{K}}_n}$ is an element from Ker $\Delta-\alpha^2$ in $\R^n \backslash \Lambda_n$.  The main point that remains to show is 
\begin{lemma} For all $k:=k_1 e_1 + \cdots + k_n e_n  \in \Lambda_n$ we have 
$$
\wp_{\alpha}^{{\cal{K}}_n}(x+k) = \wp_{\alpha}^{{\cal{K}}_n}(x_1,\ldots,x_{n-1},(-1)^{k_n} x_n).
$$ 
\end{lemma}
\begin{proof}
To prove this statement it is important to use the following decomposition  
\begin{eqnarray*}
\wp_{\alpha}^{{\cal{K}}_n}(x) &=& \sum\limits_{(m_1,\ldots,m_n) \in \mathbb{Z}^n} E_{\alpha}(x_1+m_1,\ldots,x_{n-1} + m_{n-1},(-1)^{m_n} x_n + m_n) \\
&=& \sum\limits_{(m_1,\ldots,m_{n-1}) \in \mathbb{Z}^{n-1},m_n \in 2 \mathbb{Z}} E_{\alpha}(x_1+m_1,\ldots,x_{n-1} + m_{n-1}, x_n + m_n) \\
&=& \sum\limits_{(m_1,\ldots,m_{n-1}) \in \mathbb{Z}^{n-1},m_n \in 2 \mathbb{Z}+1} E_{\alpha}(x_1+m_1,\ldots,x_{n-1} + m_{n-1}, -x_n + m_n). 
\end{eqnarray*}
First we note that 
$$
\wp_{\alpha}^{{\cal{K}}_n}(x_1+k_1,\ldots,x_{n-1}+k_{n-1}, x_n) = \wp_{\alpha}^{{\cal{K}}_n}(x_1,\ldots,x_n)
$$
for all $(k_1,\ldots,k_{n-1}) \in \mathbb{Z}^{n-1}$. 
This follows by the direct series rearrangement 
\begin{eqnarray*}
& & \wp_{\alpha}^{{\cal{K}}_n}(x_1+k_1,\ldots,x_{n-1}+k_{n-1},x_n)\\
&=& \sum\limits_{(m_1,\ldots,m_n) \in \mathbb{Z}^n} E_{\alpha}(x_1+k_1+m_1,\ldots,x_{n-1}+k_{n-1}+m_{n-1},(-1)^{m_n} x_n + m_n)\\
&=& \sum\limits_{(p_1,\ldots,p_n) \in \mathbb{Z}^n} E_{\alpha} (x_1+p_1,\ldots,x_{n-1} + p_{n-1},(-1)^{p_n} x_n+p_n)
\end{eqnarray*}
where we put $p_i:= m_i + k_i \in \mathbb{Z}$ for $i=1,\ldots,n-1$ and $p_n := m_n$. Notice that rearrangement is allowed because the series converges normally on $\R^n\backslash \Lambda_n$. 
\par\medskip\par
It thus suffices to show 
$$
\wp_{\alpha}^{{\cal{K}}_n}(x_1,\ldots,x_{n-1},x_n+1)= \wp_{\alpha}^{{\cal{K}}_n}(x_1,\ldots,x_{n-1},-x_n).
$$
We observe that 
\begin{eqnarray*}
& &\wp_{\alpha}^{{\cal{K}}_n}(x_1,\ldots,x_{n-1},x_n+1) \\
&=& \sum\limits_{(m_1,\ldots,m_n) \in \mathbb{Z}^n} E_{\alpha}(x_1+m_1,\ldots,x_{n-1}+m_{n-1},(-1)^{m_n}(x_n+1)+m_n)\\
&=&  \sum\limits_{(m_1,\ldots,m_{n-1}) \in \mathbb{Z}^{n-1},m_n \in 2 \mathbb{Z}} E_{\alpha}(x_1+m_1,\ldots,x_{n-1}+m_{n-1}, x_n+\underbrace{m_{n}+1}_{odd})\\
&+&  \sum\limits_{(m_1,\ldots,m_{n-1}) \in \mathbb{Z}^{n-1},m_n \in 2 \mathbb{Z}+1} E_{\alpha}(x_1+m_1,\ldots,x_{n-1}+m_{n-1}, -x_n+\underbrace{m_{n}-1}_{even})\\
&=&  \sum\limits_{(p_1,\ldots,p_{n-1}) \in \mathbb{Z}^{n-1},p_n \in 2 \mathbb{Z}+1} E_{\alpha}(x_1+p_1,\ldots,x_{n-1}+p_{n-1}, x_n+p_n)\\
&=&  \sum\limits_{(p_1,\ldots,p_{n-1}) \in \mathbb{Z}^{n-1},q_n \in 2 \mathbb{Z}} E_{\alpha}(x_1+p_1,\ldots,x_{n-1}+p_{n-1}, -x_n+q_n)\\
&=& \wp_{\alpha}^{{\cal{K}}_n}(x_1,\ldots,x_{n-1},-x_n). 
\end{eqnarray*}
The fact that 
$$
\wp_{\alpha}^{{\cal{K}}_n}(x_1,\ldots,x_{n-1},x_n+k_n) = \wp_{\alpha}^{{\cal{K}}_n}(x_1,\ldots,x_{n-1},(-1)^{k_n} x_n)  
$$
is true for all $k_n \in \mathbb{Z}$ now follows by a direct induction argument on $k_n$. 
\end{proof}
With this property we may infer that $\wp_{\alpha}^{{\cal{K}}_n}$ descends to a well defined section on ${\cal{K}}_n$ by applying the projection $p_*(\wp_{\alpha}^{{\cal{K}}_n})$. The result of this projection will be denoted by ${E'}_{\alpha}(x')$. ${E'}_{\alpha}(x')$ is the canonical skew symmetric periodization of $E_{\alpha}(x)$ that is constructed in such a way that it descends to the manifold. Therefore, the reproduction property of ${E'}_{\alpha}(x'-y')$ on ${\cal{K}}_n$ follows from the reproduction property of the usual Green's kernel $E_{\alpha}(x-y)$ in Euclidean space, where we apply the usual Green's integral formula for the Klein-Gordon operator.  
\par\medskip\par
{\bf Remarks}. The fundamental solution of the Klein-Gordon operator on the usual Klein bottle in two real variables (for pinor sections with values in the trivial bundle) has the form 
$$ p_*\Bigg(\sum\limits_{v   \in \Lambda_{2}} E_{\alpha}((x_1 + m_1),(-1)^{m_2}x_2+m_2)  \Bigg).$$
In terms of this formula for the fundamental solution of the Klein-Gordon operator on the manifolds ${\cal {K}}_n$ we can deduce similar representation formulas for the solutions to the general inhomogeneous Klein-Gordon problem with prescribed boundary conditions on these manifolds as presented at the end of the previous section in the context of the M\"obius strips.     
\par\medskip\par
The fact that the manifolds ${\cal{K}}_n$ are compact manifolds has some interesting special function theoretical consequences. We shall see that we can express any arbitrary solution of the Klein-Gordon equation with unessential singularities on these manifolds as a finite sum of linear combinations of the fundamental solution ${E'}_{\alpha}$ and its partial derivatives. To proceed in this direction, we start to establish
\begin{lemma}
Let $\alpha \neq 0$. Suppose that $f: \mathbb{R}^n \to \mathbb{C}$ is an entire solution of $(\Delta-\alpha^2)f=0$ on the whole $\R^n$. If $f$ additionally satisfies 
\begin{equation}\label{Kinv}
f(x_1+m_1,\cdots,x_n+m_n) = f(x_1,\ldots,x_{n-1},(-1)^{m_n} x_n)
\end{equation}
for all $(m_1,\ldots,m_n)\in \mathbb{Z}^n$, then $f$ vanishes identically on $\R^n$.   
\end{lemma}
\begin{proof}
Since $f$ satisfies the relation $$f(x_1+m_1,\cdots,x_n+m_n) = f(x_1,\ldots,x_{n-1},(-1)^{m_n} x_n),$$ it takes all its values in the $n$-dimensional period cell $[0,1]^{n-1}\times [0,2]$, because  
$$f(x_1+m_1,x_2+m_2,\ldots,x_{n-1}+m_{n-1},x_n+2m_n)=f(x_1,x_2,\ldots,x_{n-1},x_n)$$ for all $(m_1,\ldots,m_n) \in \mathbb{Z}^n$. 
The set $[0,1]^{n-1} \times [0,2]$ is compact. Since $f$ is an entire solution of $(\Delta-\alpha^2)f=0$ on the whole $\R^n$, it is in particular continuous on $[0,1]^{n-1}\times [0,2]$. Consequently, $f$ must be bounded on $[0,1]^{n-1}\times [0,2]$ and therefore it must be bounded over the whole $\R^n$, too. 
\par\medskip\par
Since $f$ is an entire solution of $(\Delta-\alpha^2)f=0$, it can be expanded into a Taylor series of the following form, compare with \cite{ConKraMMAS2009}, 
$$
f(x) = \sum\limits_{q=0}^{\infty} |x|^{1-q-n/2} J_{q+n/2-1}(\alpha|x|) H_q(x).
$$
This Taylor series representation holds in the whole space $\R^n$. 
\par\medskip\par
Here, $H_q(x)$ are homogeneous harmonic polynomials of total degree $q$. These are often called spherical harmonics, cf. for example~ \cite{M}.  

Since the Bessel $J$ functions are exponentially unbounded away from the real axis, $f$ can only be bounded if all spherical harmonics $H_q$ vanish identically. Hence, $f \equiv 0$.   
\end{proof}
Notice that all constant functions $f \equiv C$ with $C \neq 0$ are not solutions of $(\Delta-\alpha^2)f=0$. As a direct consequence we obtain
\begin{corollary}
There are no non-vanishing entire solutions of $(\Delta-\alpha^2)f=0$ on the manifolds ${\cal{K}}_n$ (in particular on the Klein bottle ${\cal{K}}_2$). 
\end{corollary}
This is a fundamental consequence of the compactness of the manifolds ${\cal{K}}_n$. Notice that this argument cannot be carried over to the context of the manifolds that we considered in the previous section, since those are not compact. 
\par\medskip\par
{\bf Remark}. The statement can be adapted to the harmonic case $\alpha=0$. In this case one has a Taylor series expansion of the simpler form 
$$
f(x) = \sum\limits_{q=0}^{\infty} H_q(x),
$$
where only the spherical harmonics of total degree $q=0,1,\ldots$ are involved. The only bounded entire harmonic functions are constants. Applying the same argumentation leads to the fact that the only harmonic solutions on ${\cal{K}}_n$ are constants.  
\par\medskip\par
We can say more. We shall now explain how we can use the function $\wp^{{{\cal K}}_n}(x)$ as basic function to construct any arbitrary solution to the Klein-Gordon operator on ${\cal{K}}_n$. 
\par\medskip\par
In order to make the paper self-contained we briefly recall from \cite{KraHabil,Nef2} the classification of the singularities and the definition of the order of singularities, adapted here to the context of harmonic functions in $\R^n$. We start with 
\begin{definition} 
A point $\tilde{x} \in \R^n$ is called a left regular point of a harmonic function $f$, if there exists an $\varepsilon > 0$ such that $f$ is harmonic in the open ball $B(\tilde{x},\varepsilon):=\{x \in \R^n \mid |x-\tilde{x}|_2 < \varepsilon\}$.  
Otherwise, $\tilde{x}$ is called a singular point of $f$. $\tilde{x}$ is called an isolated singularity of $f$ if it is a singular point of $f$ for which one can find an $\varepsilon > 0$ such that $f$ is harmonic in $B(\tilde{x},\varepsilon)\backslash\{\tilde{x}\}$.
\end{definition}
Next suppose that $U \subset \R^n$ is an open set and that $S \subset U$ is a closed subset. Let us consider around each $s \in S$ a ball  with radius $\rho > 0$ and let us denote by $H_{\rho}$ the hypersurface of  the union of all  balls centered at each $s \in S$ with radius $\rho$. In the case where $S$ is just a single isolated point, $H_{\rho}$ is simply the sphere centered at $s$ with radius $\rho$. If $S$ is a rectifiable line, or, more generally, a $p$-dimensional  manifold with boundary, then $H_{\rho}$ is the surface of a tube domain.   
Let us denote by $J(\rho)$ the limit inferior of the volumes of all closed orientable hypersurfaces $H_{R}$ with continuous normal field that contains $S$ in the interior where we suppose that 
$$
\inf\{|s-\tilde{x}|_2; s \in S, \tilde{x} \in H_{R} \} \ge \rho.
$$
In terms of these notions one can now give the following definition which was given in \cite{KraHabil} for Clifford algebra valued  monogenic functions.
\begin{definition}
Let $U \subset  \R^n$ be an open set and let $S \subset U$ be a closed subset. Suppose that $f$ is harmonic in $U \backslash S$ and that $f$ has singularities in each $s \in S$. Let $H_{\rho} \subset U$ be a hypersurface as defined above. 
The singular point $s \in S$ is then called an unessential singularity of $f$ if there is an $r \in \mathbb{N}$ and an $M>0$ such that
\begin{equation}\label{Hrho}
|\rho^r J(\rho)  f(\tilde{x}) |_2< M  \quad \quad \tilde{x} \in H_{\rho}.
\end{equation}
In this case the minimum of the parameters $r$ is called the order of the singular point $s$.  If no finite $r$ with this property can be found, then $s$ is called an essential singularity.
\end{definition}
In the case where $S$ is a isolated point singularity, a rectifiable line or a manifold with boundary, one can substitute the condition (\ref{Hrho}) by 
$$
|\rho^r f(\tilde{x}) |_2 < M.
$$
\par\medskip\par
Suppose now that $f$ is a function with the invariance behavior (\ref{Kinv}) that satisfies $(\Delta - \alpha^2) f = 0$ in $\R^n
\backslash \Lambda_n$. Let ${\bf m} := (m_1,\ldots,m_n) \in \mathbb{N}_0^n$ be
a multi-index with length $|{\bf m}|:= m_1 + \cdots + m_n$. Then
the function $f_{\bf m}(x) = \frac{\partial^{|{\bf
m}|}}{\partial x_1^{m_1} \cdots \partial x_n^{m_n}} f(x)$ satisfies (\ref{Kinv}) too, and it is also in Ker $(\Delta -
\alpha^2)$. In particular, when taking the function $$\wp_{\alpha,q}^{{\cal{K}}_n}(x) :=
\sum\limits_{v   \in \Lambda_l \oplus \Lambda_{n-l}} (-1)^{m_1+\cdots+m_l} E_{\alpha}(\sum_{i=1}^{n-1} (x_i + m_i) e_i,(-1)^{m_n}x_n+x_m)
,$$
then the functions $\wp_{\alpha,q;{\bf m}}^{{\cal{K}}_n}(x) :=
\frac{\partial^{|{\bf m}|}}{\partial x_1^{m_1} \cdots \partial x_n^{m_n}}
\wp_{\alpha,q}^{{\cal{K}}_n}(x)$ have the invariance behavior  (\ref{Kinv}) 
and satisfy $(\Delta-\alpha^2) \wp_{\alpha,q;{\bf m}}^{{\cal{K}}_n}(x) = 0$ at 
each point of $\R^n \backslash \Lambda_n$. In each lattice point they 
have an isolated pole of order $n-2+|{\bf m}|$. In view of the
translation invariance of the operator $\Delta-\alpha^2$, we can
construct functions satisfying $(\ref{Kinv})$ that have poles in a given
set of points $a_i + \Lambda_n$ of order $N_i$ ($i=1,\ldots,l)$ with
$N_i \ge n-2$ by making the construction
\begin{equation}
\label{con} \sum\limits_{i=1}^l \wp_{\alpha,q;{\bf N}_i}^{{\cal{K}}_n}(x -
a_i) b_i
\end{equation}
where ${\bf N}_i$ is a multi-index of length $N_i$ and where $b_i$
are arbitrary elements from  $\mathbb{C}$. 
Due to the compactness of the fundamental period cell, one can only construct functions in Ker $\Delta - \alpha^2$ satisfying (\ref{Kinv}) with a finite number of isolated singularities.  

In contrast to the classical harmonic case, it is possible to construct non-vanishing sections on the manifolds ${\cal{K}}_n$ in the kernel of the Klein-Gordon operator with only one point singularity of order $n-2$ on the whole manifold. The most important example is  ${E'}_{\alpha,q}(x')$. 

Now we may formulate two main representation theorems that fully describe all sections on ${\cal{K}}_n$ in the kernel of the Klein-Gordon operator with non-essential singularities. 

\begin{theorem}\label{th2} 
Let ${a}_1,{a}_2,\ldots,{a}_p \in \R^n \backslash \Lambda_n$ be a
finite set of points that are incongruent modulo $\Lambda_n$. Suppose 
that $f: \R^n \backslash\{{a}_1+\Lambda_n,\ldots,{a}_p+\Lambda_n\} \to 
{\mathbb{C}}$ is an element from Ker $\Delta-\alpha^2$ satisfying (\ref{Kinv})  which has at most isolated poles at
the points ${a}_i$ of the order $K_i$ (where the order is defined as in \cite{KraHabil}). Then there exist 
numbers $b_1,\ldots,b_p \in \mathbb{C}$ such that
\begin{equation}f(x)= \sum\limits_{i=1}^p \sum\limits_{m=0}^{K_i-(n-2)}\sum\limits_{m=m_1+m_2+\cdots+m_n}
\Bigg[\wp_{\alpha,q;{\bf m}}^{{\cal{K}}_n}({\bf x}-{a}_i)b_i\Bigg].
\end{equation}
\end{theorem}
\begin{proof} Since $f$ is supposed to be in Ker $\Delta-\alpha^2$ 
with isolated poles at the points $a_i$ of order $K_i$ its singular 
parts of the local Laurent series expansions are of the form
$E_{\alpha,q;{\bf m}}(x-{a}_i)b_i$ at each point $a_i +
\Omega$, where $E_{\alpha,q,{\bf m}}(y):=
\frac{\partial^{|{\bf m}|}}{\partial y_1^{m_1} \cdots \partial y_n^{m_n}}  
E_{\alpha,q}(y)$. As a sum of functions with the properties of satisfying $(\ref{Kinv})$ and of being in Ker $\Delta-\alpha^2$ , the expression
$$
g(x) = \sum\limits_{i=1}^p \sum\limits_{m=0}^{K_i-(n-2)}\sum\limits_{m=m_1+m_2+\cdots+m_n} \Bigg[\wp_{\alpha,q,{\bf m}}^{{\cal{K}}_n}(x-{a}_i)b_i\Bigg]
$$
also satisfies (\ref{Kinv}) and has also the same principal parts as
$f$. Hence, the function $h:=g-f$ also satisfies (\ref{Kinv}) and is in Ker $\Delta-\alpha^2$. However, $h$ has no singular parts, since
these are canceled out. The function $h$ is hence an entire solution in Ker $\Delta-\alpha^2$ which satisfies (\ref{Kinv}) and must
therefore vanish as a consequence of the preceding lemma.
\end{proof}
We can adapt this theorem to the case of dealing with functions satisfying (\ref{Kinv}) that have non-isolated singularities in the
following way:
\begin{theorem}\label{repthm2}
Suppose that $f$ is a $\mathbb{C}$-valued function satisfying (\ref{Kinv}) and $(\Delta- \alpha^2) f = 0$ in $\R^n$ except in a finite number of components of non-isolated singular sets $S_1,\ldots,S_l$ of the
orders $N(S_1),\ldots,N(S_l)$ (as defined in \cite{KraHabil}) which
are supposed to be incongruent modulo $\Lambda_n$. Then there exists
functions $b_{\bf j}:S_j \to \mathbb{C}$ of bounded variation
such that
$$
f(x) = \sum\limits_{i=1}^l \sum\limits_{{\bf j} \in {\bf J}^{(i)}} \Big(\int_{S_i} \wp_{\alpha,q;{\bf j}}^{{\cal{K}}_n}(x - c^{(i)}) d[b_{\bf j}(c^{(i)})]\Big),
$$
where the integral has to be understood as Lebesgue-Stieltjes
integral as in \cite{KraHabil,Nef2}. Here we denote by ${\bf
J}^{(i)}$ the set of those multi indices ${\bf j}$ for which the
functions $b_{\bf j}^{(i)}(c^{i})$ do not vanish identically.
\end{theorem}
To establish this more general version one can adapt the arguments
of the proof to the previous theorem and the arguments using
Lebesgue-Stieltjes integrals as applied in the context of
monogenic $n$-fold periodic functions with non-isolated
singularity sets in \cite{KraHabil}, Chapter 2.
\par\medskip\par
From Theorem~\ref{repthm2} we may now readily derive by applying the projection linear map $p_*$ the following statement.
\begin{theorem}\label{repthm3}
Let $S_1',\ldots,S_l' \subset {\cal{K}}_n$ be closed subsets.  
Suppose that $f'$ is an ${\cal{E}}^{(q)}$-valued section on ${\cal{K}}$ that satisfies the Klein-Gordon equation $\Delta'_{\alpha} f'=0$ on the manifold except in a finite number of components of non-isolated singular sets $S_1',\ldots,S_l'$ of the
orders $N(S_1'),\ldots,N(S_l')$.  Then there exist sections $b_{\bf j}':S_j' \to {\cal{E}}^{(q)}$ of bounded variation
such that
$$
f'(x') = \sum\limits_{i=1}^l \sum\limits_{{\bf j} \in {\bf J}^{(i)}} \Big(\int_{S_i'} p_*\Big(\wp_{\alpha,q;{\bf j}}^{{\cal{K}}_n}(x' - {c'}^{(i)})\Big) d[b_{\bf j}'({c'}^{(i)})]\Big),
$$
where we use the same notation as in Theorem~\ref{repthm2}. 
\end{theorem}
In the case of dealing with essential singularities the latter series is extended over infinitely many multi-indices ${\bf j}$. 
\section{Acknowledgements} The author is very thankful to Dr. Denis Constales from Ghent University for the fruitful discussions on the asymptotics of the fundamental solution of the Klein-Gordon operator in $\R^n$.  

\end{document}